%% file: main.tex
\documentclass[a4paper,autoref]{lipics-v2021}
\hideLIPIcs
\nolinenumbers
\usepackage{amsmath,amssymb,amsfonts,latexsym}
\usepackage{color,enumerate,graphicx}
\usepackage{url,hyperref}
\usepackage{fancybox}
\usepackage{algorithm}
\usepackage[noend]{algpseudocode}
\usepackage{mathtools}
\usepackage{xspace}
\definecolor{light-gray}{gray}{0.95}
\floatstyle{plaintop}
\newcommand{\eps}{\varepsilon}
\newcommand{\etal}{\emph{et al.}\xspace}

\theoremstyle{plain}

\newcommand{\C}{\ensuremath{\mathcal{C}}}
\newcommand{\D}{\ensuremath{\mathcal{D}}}

\newcommand{\G}{\ensuremath{\mathcal{G}}}

\newcommand{\T}{\ensuremath{\mathcal{T}}}
\newcommand{\U}{\ensuremath{\mathcal{U}}}

\newcommand{\REAL}{\ensuremath{\mathbb{R}}}
\newcommand{\Reals}{\REAL}

\renewcommand{\leq}{\leqslant}
\renewcommand{\geq}{\geqslant}
\newcommand{\bd}{\partial}

\DeclarePairedDelimiter\ceil{\lceil}{\rceil}

\DeclareMathOperator{\polylog}{polylog}

\DeclareMathOperator{\radius}{radius}

\newcommand{\myroot}{\mathit{root}}

\newcommand{\graph}{\G}
\newcommand{\tg}{\graph_{\mathrm{tr}}}
\newcommand{\ig}{\graph^{\times}}
\newcommand{\sep}{\mathcal{S}}

\newcommand{\reaches}{\rightsquigarrow}

\newcommand{\myin}{\mathrm{in}}
\newcommand{\myout}{\mathrm{out}}

\newcommand{\myhop}{\mathrm{hop}}
\newcommand{\dhop}{d_{\myhop}}
\newcommand{\mypath}{L}
\newcommand{\MinIn}{\mathrm{MinIn}}
\newcommand{\MaxOut}{\mathrm{MaxOut}}

\newcommand{\Bin}{B_{\myin}}
\newcommand{\Bout}{B_{\myout}}
\newcommand{\Flarge}{F_{\mathrm{large}}}

\title{A Note on Reachability and  Distance Oracles for Transmission Graphs}

\funding{MdB is supported by the  Dutch Research Council (NWO) through    Gravitation-grant NETWORKS-024.002.003.}
\author{Mark de Berg}{Department of Mathematics and Computer Science, TU Eindhoven, the Netherlands}{M.T.d.Berg@tue.nl}{}{}
\authorrunning{M.~de Berg} %mandatory. First: Use abbreviated first/middle names. Second (only in severe cases): Use first author plus 'et. al.'
\Copyright{Mark de Berg}%mandatory, please use full first names. LIPIcs license is "CC-BY";  http://creativecommons.org/licenses/by/3.0/


% mandatory: Please provide 1-5 keywords
\EventEditors{John Q. Open and Joan R. Access}
\EventNoEds{2}
\EventLongTitle{42nd Conference on Very Important Topics (CVIT 2016)}
\EventShortTitle{CVIT 2016}
\EventAcronym{CVIT}
\EventYear{2016}
\EventDate{December 24--27, 2016}
\EventLocation{Little Whinging, United Kingdom}
\EventLogo{}
\SeriesVolume{42}
\ArticleNo{23}
\begin{document}
\maketitle
%------------------------------------------------------------------------------------------

%------------------------------------------------------------------------------------------
\begin{abstract}
Let $P$ be a set of $n$ points in the plane, where each point $p\in P$ has a transmission 
radius $r(p)>0$. The transmission graph defined by $P$ and the given radii, denoted by $\tg(P)$,
is the directed graph whose nodes are the points in $P$ and that contains the arcs $(p,q)$ 
such that $|pq|\leq r(p)$. 

An and Oh [Algorithmica 2022] presented a reachability oracle for
transmission graphs. Their oracle uses $O(n^{5/3})$ storage and, given two query points $s,t\in P$,
can decide in $O(n^{2/3})$ time if there is a path from $s$ to $t$ in $\tg(P)$. 
We show that the clique-based separators introduced by De~Berg~\etal [SICOMP 2020]
can be used to improve the storage of the oracle to $O(n\sqrt{n})$ and the query time to~$O(\sqrt{n})$.
Our oracle can be extended to approximate distance queries: 
we can construct, for a given parameter~$\eps>0$, 
an oracle that uses $O((n/\eps)\sqrt{n}\log n)$ storage and that can report in
$O((\sqrt{n}/\eps)\log n)$ time a value $\dhop^*(s,t)$ satisfying 
$\dhop(s,t) \leq \dhop^*(s,t) < (1+\eps)\cdot \dhop(s,t) + 1$, 
where $\dhop(s,t)$ is the hop-distance from $s$ to $t$. 
We also show how to extend the oracle to so-called continuous queries, where the
target point~$t$ can be any point in the plane.

To obtain an efficient preprocessing algorithm, we show that a clique-based separator
of a set~$F$ of convex fat objects in $\Reals^d$ can be constructed in $O(n\log n)$
time.
\end{abstract}
%------------------------------------------------------------------------------------------

\input{introduction.tex}
\input{separators.tex}
\input{reachability-oracle.tex}
\input{conclusion.tex}

\bibliography{references}

\end{document}

%% file: introduction.tex
%--------------------------------------------------------------------------------------
\section{Introduction}
\label{sec:intro}
%--------------------------------------------------------------------------------------
Let $P$ be a set of $n$ points in the plane, where each point $p\in P$ has an associated 
\emph{transmission radius}~$r(p)>0$. The \emph{transmission graph} defined by $P$ and the 
given radii, denoted by $\tg(P)$, is the directed graph whose nodes are the points 
in~$P$ and that contains an arc~$(p,q)$ between two points $p,q\in P$ if and only if 
$|pq|\leq r(p)$, where $|pq|$ denotes the Euclidean distance between $p$ and~$q$. 
Transmission graphs are a popular way to model wireless communication networks. 

The study of transmission graphs leads to various challenging
algorithmic problems. In this paper we are interested in so-called \emph{reachability
queries}: given two query points $s,t\in P$, is there a path\footnote{Whenever we speak 
about \emph{paths} in $\tg(P)$, we always mean directed paths.} from $s$ to~$t$
in $\tg(P)$? A closely related query asks for the \emph{hop-distance} from $s$
to~$t$, which is the minimum number of arcs on any path from~$s$ to~$t$.
Reachability queries can be answered in~$O(1)$ time if we precompute 
for every pair $p,q\in P$ whether or not $q$ is reachable from~$p$ and then store that 
information in a two-dimensional array. This requires quadratic storage. 
Our goal is to develop a structure using subquadratic storage that allows us to quickly 
answer reachability queries.
Such a data structure is called a \emph{reachability oracle}. If the data structure
can also report the (approximate) hop-distance from $s$ to~$t$ then it is called an
\emph{(approximate) distance oracle}.
\medskip

Note that reachability queries for an undirected graph~$G$ can be trivially
answered in $O(1)$ time after computing the connected components of~$G$.
For directed graphs, however, efficient reachability oracles are much harder to design.
In fact, for arbitrary directed graphs one cannot hope for an oracle
that uses subquadratic storage. To see this, consider the class of directed bipartite
graphs~$G=(V\cup W,E)$ with $|V|=|W|=n/2$. Then any reachability oracle must 
use $\Omega(n^2)$ bits of storage in the worst case, otherwise two different graphs 
will end up with the same encoding and the oracle cannot answer all queries correctly on both graphs.
Surprisingly, even for sparse directed graphs no reachability oracles are known that
use subquadratic storage and have sublinear query time.
Directed planar graphs, on the other hand, admit very efficient reachability oracles: 
almost three decades of 
research~\cite{DBLP:conf/esa/ArikatiCCDSZ96,DBLP:conf/stoc/ChenX00,Djidjev1996,DBLP:journals/siamcomp/Frederickson87} 
eventually resulted in the optimal solution by Holm, Rotenberg and Thorup~\cite{DBLP:conf/focs/HolmRT15},
who presented an oracle that has $O(n)$ storage and ~$O(1)$ query time.
Moreover, there are various (approximate) distance oracles for directed
planar graphs; see for example~\cite{DBLP:conf/soda/GawrychowskiMWW18,DBLP:conf/focs/LeW21,DBLP:conf/soda/Wulff-Nilsen16}
and the references therein.

Planar graphs are sparse, but transmission graphs can be dense.
Nevertheless, the geometric structure of transmission graphs makes it possible to beat the 
quadratic lower bound on the amount of storage. 
% Kaplan~\etal~\cite{DBLP:journals/siamcomp/KaplanMRS18} were the first to study
% reachability oracles for transmission graphs. 
% \mdb{Actually, this paper only does spanners and the conversion of a reachability
% oracle to a continuous reachability oracle -- the SoCG paper was split into two.}
Kaplan~\etal~\cite{DBLP:journals/algorithmica/KaplanMRS20}
presented three reachability oracles for transmission graphs. Their oracles use subquadratic
storage when $\Psi$, the ratio between the largest and smallest  transmission radius, is sufficiently
small. The first oracle has excellent performance---it uses $O(n)$ storage 
and can answer queries in~$O(1)$ time---but it only works when $\Psi<\sqrt{3}$. The second oracle 
works for any~$\Psi$, but it uses $O(\Psi^3 n\sqrt{n})$ storage and has $O(\Psi^3 \sqrt{n})$
query time. The third oracle has a reduced dependency on $\Psi$---it has
$O((\log^{1/3}\Psi) \cdot n^{5/3} \log^{2/3}n)$ storage and $O((\log^{1/3}\Psi)\cdot n^{2/3} \log^{2/3}n)$
query time---but an increased dependency on~$n$.
This third structure is randomized and answers queries correctly with high probability.
Recently, An and Oh~\cite{An-Oh-Algorithmica} presented the first reachability oracle
that has subquadratic storage independent of~$\Psi$. It uses $O(n^{5/3})$ storage
and has $O(n^{2/3})$ query time. They also presented an oracle that uses
$O(n\sqrt{n\log \Psi})$ storage and has $O(\sqrt{n\log\Psi})$ query time.

%--------------------------------------------------------------------------------------
\subparagraph*{Our contribution and technique.}
%--------------------------------------------------------------------------------------
We present a reachability oracle for transmission
graphs that uses $O(n\sqrt{n})$ storage and has $O(\sqrt{n})$ query time. This 
significantly improves both the storage and the query time of the oracle presented 
by An and Oh~\cite{An-Oh-Algorithmica}. 
Our oracle uses a divide-and-conquer approach based on separators. This is a standard
approach for reachability oracles---it goes back to (at least) the work of Arikati~\etal~\cite{DBLP:conf/esa/ArikatiCCDSZ96}
and is also used by Kaplan~\etal and by An and Oh.
It works as follows. 

Consider a directed graph~$G=(V,E)$ and let $S\subset V$ be a separator for~$G$.
Thus $V\setminus S$ can be partitioned into two parts $A,B$ of roughly equal size 
such that there are no arcs between~$A$ and~$B$.
The oracle now stores for each pair of nodes $(u,v) \in V\times S$ whether $u$
can reach~$v$ in~$G$, and whether $v$ can reach~$u$ in~$G$. In addition, there are
recursively constructed oracles for the subgraphs $G[A]$ and $G[B]$ induced by $A$ and~$B$, respectively.
Answering a reachability query with vertices~$s,t$ can then be done as follows:
we first check if $s$ can reach~$t$ via a node in~$S$, that is,
if there is a node~$v\in S$ such that $s$ can reach~$v$ and $v$ can reach~$t$.
Using the precomputed information this takes $O(|S|)$ time. If $s$ can
reach $t$ via a node in~$S$, we are done. Otherwise, $s$ can only reach $t$ if
$s$ and $t$ lie in the same part of the partition, say~$A$, and $s$
can reach $t$ in~$G[A]$. This can be checked recursively. 

For graph classes that admit a separator of size~$s(n)$, where $n := |V|$, this approach
leads to an oracle with $O(n \cdot s(n))$ storage and $O(s(n))$ query time,
assuming $s(n)=\Omega(n^{\alpha})$ for some~$\alpha>0$.
The problem when using this approach for transmission graphs is that 
they may not have separators of sublinear size. An and Oh~\cite{An-Oh-Algorithmica}
therefore apply the following preprocessing step. Let $\D_P$ be the set
of disks defined by the points in $P$ and their ranges. Then An and Oh iteratively 
remove collections of $\Omega(n^{1/3})$ disks that contain a common point.
Note that the disks in such a collection form a clique in the intersection graph defined by~$\D_P$.
For each collection~$C$, a structure is built to
check for two query points $s,t\in P$ if $s$ can reach~$t$ via a point corresponding to a disk in~$C$.
The set of disks remaining after this preprocessing step has 
\emph{ply}~$O(n^{1/3})$---any point in the plane is contained in $O(n^{1/3})$ of them.
This implies that the corresponding transmission graph admits a separator of 
size~$O(n^{2/3})$~\cite{MTTV-sep-sphere-packing,SW-geom-sep}, leading to 
a reachability oracle with $O(n^{5/3})$ storage and $O(n^{2/3})$ query time.\footnote{Actually, this preprocessing 
needs to be done at each step in the recursive construction of the oracle. 
Otherwise the ply would be $O(n_0^{2/3})$ instead of $O(n^{1/3})$,
where $n_0$ is the initial number of points and $n$ is the number of points
in the current recursive call, resulting in an extra logarithmic factor in storage.}
Our main insight is that we can use the so-called \emph{clique-based separators} recently
developed by De~Berg~\etal~\cite{bbkmz-ethf-20}. (A clique-based separator
for a graph~$\G$ is a separator that consists of a small number of cliques,
rather than a small number of nodes.) This allows use to avoid An and Oh's
preprocessing step and integrate the handling of cliques into the global separator approach,
giving an oracle with $O(n\sqrt{n})$ storage and $O(\sqrt{n})$ query time.

We also show how to adapt the 
oracle such that it can answer approximate hop-distance queries. In particular, 
we show how to construct, for a given parameter~$\eps>0$, an approximate distance oracle that 
uses $O((n/\eps)\sqrt{n}\log n)$ storage. The oracle can report, 
for two query points $s,t\in P$, a value $\dhop^*(s,t)$ satisfying
$\dhop(s,t) \leq \dhop^*(s,t) < (1+\eps)\cdot \dhop(s,t) + 1$, where 
$\dhop(s,t)$ denotes the hop-distance from $s$ to $t$ in $\tg(P)$.
The query time is $O((\sqrt{n}/\eps)\log n)$. 

Finally, we study so-called \emph{continuous reachability queries}, where the target point~$t$ 
can be any point in $\Reals^2$. Such a query asks, for a query pair $s,t\in P\times \Reals^2$, 
if there is a point $q\in P$ with $|qt|\leq r(q)$ such that $s$ can reach~$q$.
Kaplan~\etal\cite{DBLP:journals/algorithmica/KaplanMRS20} presented a data structure
using $O(n \log \Psi)$ storage such that any reachability oracle can be extended to
continuous reachability queries\footnote{Kaplan~\etal used the term \emph{geometric reachability queries}.}
with an additive overhead of $O(\log n\log \Psi)$ to the query time.
An and Oh~\cite{An-Oh-Algorithmica} also extended their reachability oracle 
to continuous reachability queries. The storage of their oracle remained the same,
namely~$O(n^{5/3})$, but the query increased by a multiplicative polylogarithmic factor to~$O(n^{2/3}\log^2 n)$.
We present a new data structure to extend any reachability oracle to continuous
queries. It uses $O(n\log n)$ storage and has an additive overhead of $O(\log^2 n)$
to the query time. 
%  \mdb{Check. I think we can also get $O(n\log \Psi)$ storage and $O(\log n\log\Psi)$ additive overhead,
% but perhaps that's not better than what was known..} 
Thus, unlike the method of Kaplan~\etal, its performance
is independent of the ratio~$\Psi$, and unlike the method of An and Oh, we do not
incur a penalty on the query time when combined with our reachability oracle:
the total query time remains $O(\sqrt{n})$ and the total storage remains $O(n\sqrt{n})$.
The extension to continuous queries also applies to approximate distance queries,
with an additional additive error of a single hop.
%
% \mdb{Let's not do the following:}
% Finally, we show that the oracle can be generalized to the setting where the 
% transmission range of a point~$p\in P$ is a fat convex region ``centered'' at~$p$,
% and that the solution also works for transmission graphs
% in $\Reals^d$ for $d>2$; the storage and query time of the reachability oracle 
% then become $O(n^{2-1/d})$ and $O(n^{1-1/d})$, respectively
\medskip

To construct our oracles, we need an algorithm to construct a clique-based separator
for the intersection graph $\ig(\D)$ of a set~$\D$ of~$n$ disks in the plane.
De~Berg~\etal~\cite{bbkmz-ethf-20} present a brute-force construction algorithm for the case
where $\D$ is a set of convex fat objects in $\Reals^d$, running in $O(n^{d+2})$ time. This
is sufficient for their purposes, since they use the separator to develop sub-exponential
algorithms. For us the separator construction would become the bottleneck in our
preprocessing algorithm. Before presenting our oracles,
we therefore first show how to construct the clique-based separator  
in $O(n\log n)$ time. Since we expect that clique-based separators will
find more applications where the construction time is relevant,
we believe this result is of independent interest.

%% file: separators.tex
%------------------------------------------------------------------------------------------
\section{A fast algorithm to construct clique-based separators}
\label{sec:separators}
%------------------------------------------------------------------------------------------
Our distance oracle for transmission graphs will be based on the so-called clique-based
separators of De~Berg~\etal~\cite{bbkmz-ethf-20}. In this section we present an efficient
construction algorithm for these separators. Although we only need them for disks
in the plane, we will describe the construction algorithm for convex fat objects in~$\Reals^d$,
since we expect that an efficient construction of clique-based separators may find other uses.
\medskip

Let $F$ be a set of $n$ convex $\alpha$-fat objects in~$\Reals^d$, where an object
$o\subset \Reals^d$ is \emph{$\alpha$-fat} if there are concentric balls $\Bin(o)$ and $\Bout(o)$ with 
$\Bin(o)\subseteq o\subseteq\Bout(o)$ and $\radius(\Bin(o)) \geq \alpha\cdot \radius(\Bout(o))$.
We are interested in sets of objects that are $\alpha$-fat for some 
fixed, absolute constant~$\alpha>0$. We will therefore drop the parameter~$\alpha$ 
and simply speak of \emph{fat objects} from now on.

Let $\ig(F)$ be the intersection graph induced by~$F$, that is,
$\ig(F)$ is the undirected graph whose nodes correspond to the objects in $F$ and that 
has an edge between two objects $o_1,o_2\in F$ if an only if $o_1$ and $o_2$ intersect.
A \emph{$\delta$-balanced clique-based separator} of $\ig(F)$ is a 
collection~$\sep=\{C_1,\ldots,C_{|\sep|}\}$ of cliques in $\ig(F)$ 
whose removal partitions $\ig(F)$ into two parts 
of size at most $\delta n$, with no edges between the parts. Let $\gamma$ be a function that assigns
a wight to a clique, depending on its size. Then
the \emph{weight} of a separator~$\sep$ is defined as
$\sum_{C\in \sep}\gamma(|C|)$. In other words, the weight of a separator is the sum of the weights
of its constituent cliques. 
% De~Berg~\etal use the weight function $\gamma(|C|) := \log (|C|+1)$
% in their applications. In the next section we will use clique-based separators
% with a weight function $\gamma(|C|):=1$, so the weight of a separator is simply the number of constituent cliques.
%---------------------------------------------------------------------------------------
\begin{theorem}[De~Berg~\etal~\cite{bbkmz-ethf-20}]\label{thm:arbitrary-size-separator}
Let $d\geq 2$ and $\eps>0$ be constants and let $\gamma$ be a weight function such 
that $\gamma(t) = O(t^{1-1/d-\eps})$. Let $F$ be a set of $n$ convex\footnote{The result
also holds for non-convex fat objects if they are similarly sized, but we restrict our attention to
convex fat objects.} fat objects in~$\Reals^d$. Then the intersection
graph~$\ig(F)$ has a $\delta$-balanced clique-based separator~$\sep$ 
of weight $O(n^{1-1/d})$ for some fixed constant~$\delta<1$.
\end{theorem}
%-------------------------------------------------------------------------------------- 
De~Berg~\etal\cite{bbkmz-ethf-20} show that when the objects have constant complexity,
then the clique-based separator of Theorem~\ref{thm:arbitrary-size-separator} can be 
constructed in~$O(n^{d+2})$ time, with a balance factor~$\delta=6^d/(6^d+1)$.
We show that their algorithm can be implemented to run in $O(n\log n)$ time. 
Our implementation results in a somewhat larger balance factor, namely $12^d/(12^d+1)$. 
This is still a fixed constant smaller than~1, 
which is all that matters in applications of the separator theorem. 
\medskip

\noindent \emph{Step 1: Finding candidate separators.} 
Step~1 starts by finding a \emph{smallest $k$-enclosing hypercube} for~$F$---this is
a smallest hypercube that contains at least~$k$ objects from~$F$---for $k:=n/(6^d+1)$.
De~Berg~\etal do this in a brute-force manner, taking $O(n^{d+2})$ time. 
Instead, we compute a 2-approximation\footnote{There is an efficient randomized algorithm to compute 
a smallest $k$-enclosing ball for a set of points in~$\Reals^d$~\cite{DBLP:journals/jacm/Har-PeledR15}, which we
can use, but we want to obtain a deterministic algorithm.} of a smallest $k^*$-enclosing hypercube
for $k^* := n/(12^d+1)$, as follows.

Consider a smallest $k^*$-enclosing hypercube $H^*$ for $F$.
Let $q$ be an arbitrary point in~$H^*$, and let $H_q$ denote the smallest $k^*$-enclosing 
hypercube centered at~$q$. Note that $H_q$ is a 2-approximation of~$H^*$, that is,
the edge length of~$H_q$ is at most twice the edge length of~$H^*$. 
Moreover, $H_q$ can trivially be computed in $O(n)$ time, for a given~$q$.
We can find a point~$q\in H^*$ as follows.
Pick a representative point $p_o\in o$ for every object~$o\in F$ and 
compute a \emph{weak $\eps$-net} with respect to convex ranges for the set $P_F := \{p_o:o\in F\}$,
where $\eps := 1/(12^d+1)$. (A weak $\eps$-net for convex ranges
is a point set~$Q \subset \Reals^d$ such that any convex range---and, hence, any hypercube---containing 
at least $\eps n$ points from~$P_F$ also contains at least one point from~$Q$.)
Since $\eps$ and~$d$ are fixed constants, there is such an $\eps$-net
of size~$O(1)$ and it can be constructed\footnote{Using
an $\eps$-net for convex ranges instead of for hypercubes only,
is overkill. However, we did not find a reference explicitly stating
that a linear-time deterministic algorithm exists that constructs an $\eps$-net for hypercubes.}
in~$O(n)$ time~\cite{DBLP:journals/dcg/MatousekW04}.
Note that $Q$ is guaranteed to contain a point~$q\in H^*$.
Thus, for each $q\in Q$, we compute the smallest smallest $k^*$-enclosing 
hypercube~$H_q$, and we report the smallest of these hypercubes. This hypercube,
which we denote by~$H_0$, is a 2-approximation a smallest $k^*$-enclosing hypercube.

The hypercube $H_0$ is used to define a set of potential separators, as explained next.
Assume without loss of generality that the edge length of~$H_0$ is~2, so that the
edge length of a smallest $k^*$-enclosing hypercube is at least~1.
Define $H(t)$ to be the copy of $H_0$ that is scaled by a factor~$t$ with respect to the center of~$H_0$.
Following De~Berg~\etal, we define $\Flarge\subseteq F$ to be the set of objects that intersect
the interior or boundary of~$H(3)$ and whose diameter is at least~$1/4$. 
These so-called \emph{large objects} can be stabbed by $O(1)$ points, and so they define
$O(1)$ cliques. These cliques will be put into the separator (plus some more
cliques, as detailed below), so we can 
focus on $F\setminus\Flarge$, the set of remaining objects.

Note that any hypercube $H(t)$ induces a separator $\sep(t)$ in a natural way: put
the objects from  $F\setminus\Flarge$ that intersect the boundary~$\bd H(t)$ into~$\sep(t)$, 
suitably grouped into cliques, in addition to the cliques comprising~$\Flarge$ that we already put into~$\sep(t)$.
De Berg~\etal~\cite{bbkmz-ethf-20} prove that one of the separators~$\sep(t)$ with $1\leq t\leq 3$
is a separator with the desired properties. To this end they first prove that
any separator~$\sep(t)$ with $1\leq t\leq 3$ is $(6^d/(6^d+1))$-balanced. In their proof they work with a smallest $n/(6^d+1)$-enclosing
hypercube~$H_0$. The key argument is that $H(3)$ can be covered by $6^d$ unit hypercubes % of edge length~1
in such a way that any object in $F\setminus \Flarge$ is contained in one of them. Since each 
of the covering hypercubes contains at most $n/(6^d+1)$, this results in a balance factor
of~$6^d/(6^d+1)$. We can use exactly the same argument, except that our~$H_0$ is only a 2-approximation of 
a smallest $k^*$-enclosing hypercube. Thus we need $12^d$ unit hypercubes in our cover;
see Fig.~\ref{fig:sep-construction}(i).
Since $k^*=n/(12^d+1)$, this implies that our separators $\sep(t)$ % with $1\leq t\leq 3$
are~$(12^d/(12^d+1))$-balanced.
%--------------------------------------------------------------------------------------
\begin{figure}
\begin{center}
\includegraphics{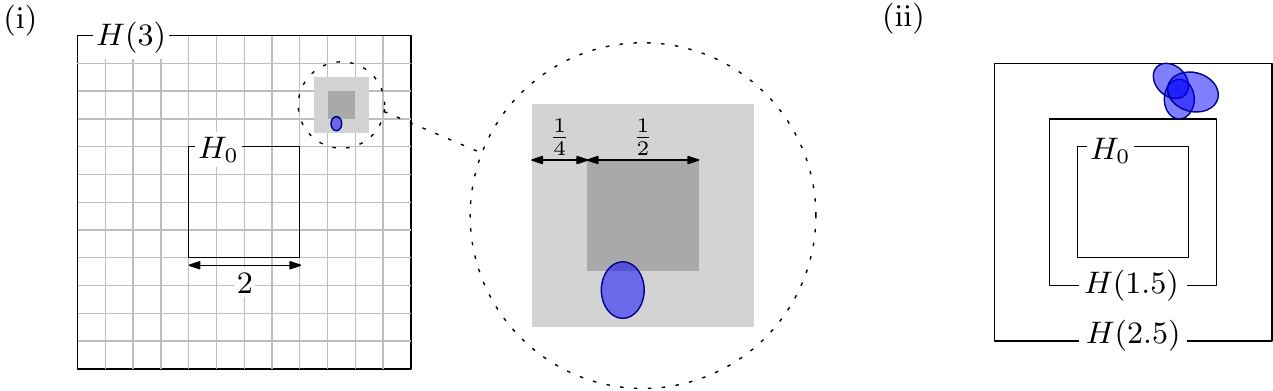}
\end{center}
\caption{(i) Hypercube~$H(3)$ is partitioned into $12^d$ cells. For each cell,
             a unit hypercube with the same center is created; see the dark grey cell and
             the light grey unit hypercube. Any object
             of diameter at most~$1/4$, such as the blue object, intersecting some cell is
             completely contained in the corresponding unit hypercube.
        (ii) For the clique~$C$ shown in blue, we have $I_C=[1.5,2.5]$.}
\label{fig:sep-construction}
\end{figure}
%--------------------------------------------------------------------------------------
\medskip

\noindent \emph{Step 2: Constructing the cliques and finding a low-weight separator.}
Step~2 starts by partitioning the set of objects that may intersect the separators~$\sep(t)$ into 
cliques. To this end the objects are partitioned into size classes,
and for each\footnote{The smallest size class is handled differently,
namely by creating a singleton clique for each object in this class. For
simplicity we do not explicitly mention these singleton cliques from now on.} 
size class $F_s$ a point set $Q_s$ of size $O(n/2^{2sd})$
is generated that stabs all objects in $F_s$. Each point
$q\in Q_s$ defines a clique~$C_q$, which consists of objects 
from $F_s$ containing~$q$; here an object stabbed by multiple points can
be assigned to an arbitrary one of them. The set $Q_s$ consists
of grid points of a suitably defined grid, and they are not only guaranteed
to stab the objects $o\in F_s$ but they even stab the inner balls $\Bin(o)$.
Thus, we can assign each object~$o\in F_s$ to the point~$q\in Q_s$ closest
to the center of~$\Bin(o)$. Since~$Q_s$ forms a grid, this can be done in 
$O(|F_s|\log n)$ time, or even in $O(|F_s|)$ time if we would allow the floor function. 
The total time to construct the set $\C$ of all cliques over all
size classes $F_s$ is therefore $\sum_s O(n/2^{2sd} + |F_s|\log n)$,
which is $O(n\log n)$ since $\sum_s |F_s|\leq n$.

De Berg~\etal prove that one of the separators ~$\sep(t)$, for $1\leq t\leq 3$
has the desired weight, that is, that the cliques from $\C$ intersected by~$\sep(t)$ have total weight~$O(n^{1-1/d})$. 
% In fact, we only put the objects intersecting~$\bd H(t)$ in the separator
% and not necessarily the whole clique.
We can find this separator in $O(n\log n)$ time, as follows. For each clique~$C\in\C$,
define 
\[ I_C := \{ t : 1\leq t\leq 3 \mbox{ and } \bd H(t)\cap \U(C) \neq \emptyset\},
\]
where $\U(C)$ denote the union of the objects comprising the clique~$C$.
In other words, $I_C\subseteq [1,3]$ contains the values of $t$ such that $\bd H(t)$
intersect at least one object from the clique~$C$; see Fig.~\ref{fig:sep-construction}(ii).
The interval $I_C$ can trivially be computed in $O(|C|)$ time, 
so computing all intervals takes $O(n)$ time in total.
We then find a value $t^*\in [1,3]$ that minimizes
$\sum_{C: t\in I_C} \gamma(|C|)$, which can easily be done in~$O(n\log n)$ time,
and report~$\sep(t^*)$ as the desired separator.  

In conclusion, both steps of the construction algorithm can be implemented
to run in $O(n\log n)$ time, leading to the following theorem.
%-------------------------------------------------------------------------------------- 
\begin{theorem}\label{thm:separator-construction}
Let $F$ be a set of $n$ constant-complexity fat objects in $\Reals^d$, where $d$ is a fixed constant.
Then we can construct a clique-based separator for $\ig(F)$ with the properties given in 
Theorem~\ref{thm:arbitrary-size-separator} in $O(n\log n)$ time. 
\end{theorem}
%-------------------------------------------------------------------------------------- 

%% file: reachability-oracle.tex
%------------------------------------------------------------------------------------------
\section{The oracles}
\label{sec:oracle}
%------------------------------------------------------------------------------------------
Before we describe our oracles in detail, we need to introduce some notation.
Recall that $P$ is the set of input points, and that $r(p)$ denotes the transmission radius 
of a point~$p\in P$. We denote the disk of radius~$r$ centered at some point~$z$ by $D(z,r)$,
and for $p\in P$ we define $D_p := D(p,r(p))$. We call $D_p$ the \emph{transmission disk}
of~$p$. Note that the arc $(p,q)$ is present in the
transmission graph~$\tg(P)$ if and only if~$q\in D_p$. We write $p\reaches q$
to indicate that $p$ can reach $q$ in~$\tg(P)$. In the following we do not
distinguish between the points in~$P$ and the corresponding nodes in $\tg(P)$.

%--------------------------------------------------------------------------------------
\subparagraph*{The basic reachability oracle.}
%-------------------------------------------------------------------------------------- 
Let $\mypath=q_1,q_2,\ldots,q_{|L|}$  be a path in $\tg(P)$. The first ingredient
that we need is an oracle for what we call a \emph{via-path query} with respect to
the given path~$L$: given two query points~$s,t\in P$, 
is it possible for $s$ to reach $t$ via a node in $\mypath$? 
In other words, a via-path query asks if there exist a node~$q_i\in \mypath$ 
such that $s\reaches q_i \reaches t$. We call such an oracle a \emph{via-path oracle}.
As observed by An and Oh~\cite{An-Oh-Algorithmica}, there is a simple via-path oracle 
with $O(1)$ query time that uses $O(n)$ storage, irrespective 
of the length of the path~$L$. The oracle stores for each point $p\in P$,
including the points in~$\mypath$, two values:
\[
\MinIn_{\mypath}[p] := \min \{ i : p \reaches q_i \mbox{ and } q_i \in \mypath \}   \mbox{ and } 
\MaxOut_{\mypath}[p] := \max \{ i : q_i \reaches p \mbox{ and } q_i \in \mypath \}.
\]
In other words, $\MinIn_{\mypath}[p]$ is the minimum\footnote{We define the minimum over the
empty set to be $\infty$, which means that $\MinIn_{\mypath}[p]=\infty$ if
no $q_i$ can be reached from~$p$. Similarly,  the maximum over the
empty set is defined to be~$-\infty$, so $\MaxOut_{\mypath}[p]=-\infty$ if
no $q_i$ can reach~$p$.} 
index of any point $q_i\in \mypath$ that can be
reached from~$p$, where the points in $\mypath$ are numbered in order along the path.
Note that the path from $p$ to $q_{i^*}$, where if $i^* := \MinIn_{\mypath}[p]$,
may pass through other points from $\mypath$ (which then must be successors of $q_{i^*}$ along~$\mypath$).
A similar statement holds for the path from $q_{j^*}$ to $p$ where $j^* := \MaxOut_{\mypath}[p]$.
Clearly the oracle needs $O(n)$ storage. It is easy to see that a via-path query with points $s,t$ can be answered by checking
if $\MinIn_{\mypath}[s] \leq \MaxOut_{\mypath}[t]$; see the paper
by An and Oh~\cite[Lemma 13]{An-Oh-Algorithmica}.

An and Oh show that a
via-path oracle can be constructed in $O(n)$ time, provided a
linear-size spanner of $\tg(P)$ is available;
see the proof of Lemma~14 in their paper.
(A spanner for $\tg(P)$ is a subgraph such that for any two points~$p,q\in P$ we have:
$p$ can reach $q$ in $\tg(P)$ if and only if $p$ can reach $q$ in the subgraph.)
% where they show that, for a given path~$\mypath$, the indices $\MinIn_{\mypath}[p]$ and $\MaxOut_{\mypath}[p]$ for all $p\in P$
% can be computed in $O(n)$ time in total.
% This is done using a BFS on a spanner of $\tg(P)$, 
They also show that a spanner can be computed in $O(n\log^3n)$ time~\cite[Theorem 5]{An-Oh-Algorithmica}.
Note that the spanner needs to be computed only once, so its construction
does not influence the preprocessing time of our final oracle 
as stated later in Theorem~\ref{thm:reachability-oracle}.
% NOTE ABOUT THE PAPER BY AN AND OH:
%
% In Section 3.2 they write "we apply the BFS algorithm in Sect. 2 from s" but this should just be
% "we apply BFS on the 2-spanner" . Otherwise it would be too slow (see tgheir Thm 5)
% and reversing the edges would be a problem
\medskip

Now let $\D_P := \{ D_p : p\in P\}$ be the set of disks defined by the points in $P$ and
their transmission radii. Let $\ig(\D_P)$ denote the intersection graph of $\D_P$,
that is, the undirected graph with node set $\D_P$ that contains an edge $(D_p,D_q)$
if and only if $D_p\cap D_q\neq \emptyset$. Let $C$ be a clique in $\ig(\D_P)$ such that all
disks in $C$ can be stabbed by a single point~$x_C$. We call such a clique~$C$, 
a \emph{stabbed clique}. A crucial observation by An and Oh~\cite{An-Oh-Algorithmica}
is that the set $P(C)\subseteq P$ of points corresponding to a stabbed clique~$C$ can be 
covered by six paths in $\tg(P)$. More precisely, $P(C)$ can be partitioned into
(at most) six subsets $P_1(C),\ldots,P_6(C)$ with the following 
property: if we order the points in $P_i(C)$ by decreasing transmission radius
then there is an arc from each point to its direct successor, so the
points form a path in~$\tg(P)$; see Figure~\ref{fig:cones}.
In fact, and this will be relevant later, each point in $P_i(C)$ has an arc to all its 
successors in the ordering. We call such a path a \emph{transitive path}. 
(An and Oh use the term~\emph{chain}.)
%--------------------------------------------------------------------------------------
\begin{figure}
\begin{center}
\includegraphics{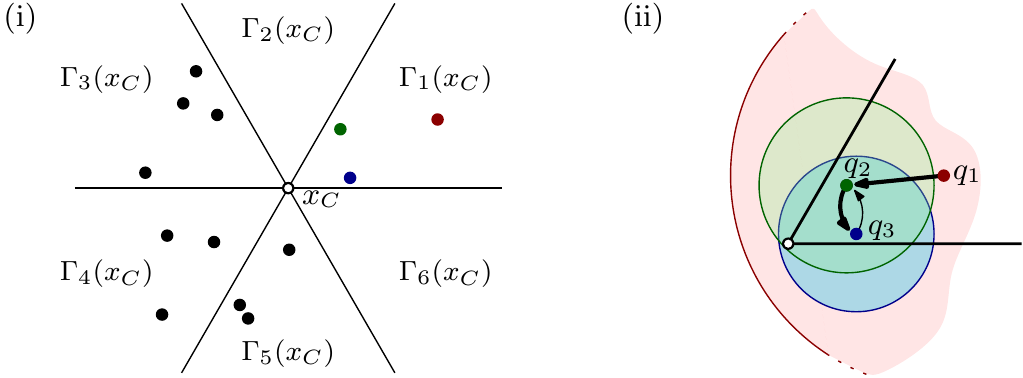}
\end{center}
\caption{(i) To partition $P(C)$ into $P_1(C),\ldots,P_6(C)$ we partition
              the plane into six 60-degree cones $\Gamma_1(x_C),\ldots,\Gamma_6(x_C)$
              with apex at $x_C$---we
              call these the \emph{canonical cones} of~$x_C$---and assign 
              each point from $P(C)$ to the cone containing it, with ties broken arbitrarily. 
              The set $P_i(C)$ then contains the points assigned to $\Gamma_i(x_C)$.
              In this example the sets $P_2(C)$ and $P_6(C)$  are empty.
         (ii) If we sort the points assigned to a cone by increasing transmission radius,
              then each point $q_i$ must have an arc to any successor $q_j$ in the ordering. 
              The reason is that the angle $\angle q_i x_C q_{j}$ is at most 60~degrees, 
              so $|q_i q_{j}| \leq \max(|q_i x_C|,|q_{j} x_C|) \leq \max(r(q_i),r(q_{j}))=r(q_i)$.}
\label{fig:cones}
\end{figure}
%--------------------------------------------------------------------------------------

It is well known and easy to see that any clique in $\ig(\D_P)$ can be partitioned
into $O(1)$ stabbed cliques. Thus we obtain the following lemma.
%-------------------------------------------------------------------------------------- 
\begin{lemma}\label{lem:cover-clique-by-paths}
Let $C$ be a clique in $\ig(\D_P)$ and let $P(C)\subseteq P$ be the points corresponding to the
disks in~$C$. Then we can cover the points in $P(C)$ by $O(1)$ transitive paths in $\tg(P)$.
\end{lemma}
%-------------------------------------------------------------------------------------- 
A \emph{via-clique oracle}, for a given clique~$C$ in $\ig(\D_P)$, is a data structure
that can answer \emph{via-clique queries}: given two query points~$s,t$,
decide if $s$ can reach $t$ in $\tg(P)$ via a point in~$P(C)$. 
It is important to note that the query asks if $s$ can reach $t$ in $\tg(P)$, not in~$\ig(\D_P)$.
Lemma~\ref{lem:cover-clique-by-paths} implies that we can answer a via-clique
query using $O(1)$ via-path queries.

The following lemma summarizes
(and slightly extends, since they only considered stabbed cliques)
the intermediate result from An and Oh on which we build.
%-------------------------------------------------------------------------------------- 
\begin{lemma}\label{lem:vc-query}
Let $C$ be a clique in $\ig(\D_P)$. Then there is a via-clique oracle for $C$ that uses 
$O(|C|)$ storage and that can answer via-clique queries in $O(1)$ time. 
The oracle can be constructed in $O(n)$ time, provided we have a spanner of the
transmission graph~$\tg(P)$ available.
\end{lemma}
%-------------------------------------------------------------------------------------- 
To obtain a reachability oracle for transmission graphs we combine 
Lemma~\ref{lem:vc-query} with the machinery of clique-based separators.
Recall from the previous section that a $\delta$-balanced clique-based separator 
of the intersection graph $\ig(\D)$ of a set~$\D$ of $n$ disks, is a set $\sep$ 
of cliques in~$\ig(\D)$ whose removal partitions $\ig(\D)$ into two parts of size at most $\delta n$, 
for some  fixed constant $\delta<1$. Note that if we set the weight function $\gamma$
to be $\gamma(|C|)=1$, then the weight of a clique-based separator is simply its number of cliques.
Theorem~\ref{thm:separator-construction} thus implies that we can compute a clique-based
separator consisting of $O(\sqrt{n})$ cliques in $O(n\log n)$ time.
Clique-based separators provide a simple mechanism to recursively construct  an efficient
reachability oracle for transmission graphs, as follows.

\begin{quotation}
\noindent
Construct a clique-based separator~$\sep$ consisting of $(\sqrt{n})$ cliques for the
intersection graph~$\ig(\D_P)$, where $\D_P$ is the set of disks defined 
by the points in~$P$ and their transmission radii.
%--------------------------------------------------------------------------------------
\begin{itemize}
\item For each clique $C\in \sep$, construct a via-clique oracle using Lemma~\ref{lem:vc-query}.
\item Let $P_A,P_B\subset P$ be the two subsets of points corresponding to the 
       partition of $\ig(\D_P)$ induced by~$\sep$. Recursively construct reachability oracles for
       $P_A$ and~$P_B$.
\end{itemize}
%--------------------------------------------------------------------------------------
The recursion ends when $|P|\leq 1$; we then simply store the point in~$P$ (if any).
\end{quotation}

Answering a reachability query with query points $s$ and $t$ is done as follows.
We first perform a via-clique query with $s$ and $t$ for each clique $C\in\sep$. 
If any of the queries returns {\sc yes} then $s$ can reach~$t$. If all queries return {\sc no},
and $s$ and $t$ belong to different parts of the partition, then 
$s$ cannot reach~$t$. If neither of these two cases apply, then we
recursively check if $s$ can reach $t$ in $\tg(P_A)$ or $\tg(P_B)$, depending
on whether $s,t\in P_A$ or $s,t\in P_B$.
This leads to the following theorem.
%-------------------------------------------------------------------------------------- 
\begin{theorem}\label{thm:reachability-oracle}
Let $P$ be a set of $n$ points in~$\Reals^2$. Then we can answer reachability queries
on~$P$ in $O(\sqrt{n})$ time with an oracle using $O(n\sqrt{n})$ storage, which 
can be constructed in $O(n\sqrt{n})$ time. 
\end{theorem}
%-------------------------------------------------------------------------------------- 
\begin{proof}
Observe that a separator in $\ig(\D_P)$ is also a separator in $\tg(P)$.
In other words, if $P_A,P_B$ are the two subsets of $P$ corresponding to
the parts in the partition of $\ig(\D_P)$ induced by~$\sep$, then there are
no arcs between the points in~$P_A$ and those in~$P_B$. This implies that
queries are answered correctly.

The total storage for the via-clique oracles of the cliques~$C\in\sep$
is $O(n\sqrt{n})$.
Hence, the storage $M(n)$ of our oracle satisfies the recurrence
$M(n) = O(n\sqrt{n}) +  M(n_1) + M(n_2)$,
where $n_1,n_2\leq \delta n$ for some constant $\delta<1$ and $n_1+n_2 \leq n$. 
It follows that $M(n) = O(n\sqrt{n})$. The query time $Q(n)$ satisfies
$Q(n) = O(\sqrt{n}) +  Q(n_1)$ where $n_1\leq \delta n$, 
and so $Q(n) = O(\sqrt{n})$. 

It remains to discuss the preprocessing time. By Lemma~\ref{lem:vc-query} 
the time to construct a via-clique oracles is~$O(n)$, after $O(n\log^3 n)$ preprocessing
to construct a spanner of~$\tg(P)$,
which only needs to be done once. Hence, the total
time to construct a via-clique oracle for each of the $O(\sqrt{n})$ cliques
is $O(n\sqrt{n})$. Moreover, by
Theorem~\ref{thm:separator-construction} the clique-based separator can be constructed
in $O(n\log n)$ time. Thus the preprocessing time satisfies the same recurrence as
the amount of storage, which implies that the total construction time is $O(n\sqrt{n})$.
\end{proof}
%-------------------------------------------------------------------------------------- 

%-------------------------------------------------------------------------------------- 
\subparagraph{An approximate distance oracle.}
%-------------------------------------------------------------------------------------- 
We now extend our reachability oracle to an approximate distance oracle.
Thus the oracle must be able to approximate the hop distance\footnote{We define 
$\dhop(s,t)=\infty$ if $s$ cannot reach $t$; in this case the oracle should return~$\infty$.}
$\dhop(s,t)$ for two query points~$s,t\in P$. 
%The extension of our via-path oracle needs the path~$\mypath$ to be
%transitive, that is, if $\mypath = q_1,\ldots,q_{|\mypath|}$     
%then we require that each $q_i\in \mypath$ has an arc to all $q_j \in\mypath$
%with $j>i$, and not just to its direct successor~$q_{i+1}$.
%Recall that the paths for which we need a via-path oracle---namely the paths
%induced by the points inside a 60-degree cone where the points
%are sorted on their transmission radius---have this property.

To answer approximate distance queries we need to extend the via-path oracle so that,
for a given transitive path $\mypath$ and query points~$s,t\in P$, it can approximate
$d_{\mypath}(s,t)$, the length of the shortest path 
from $s$ to $t$ via a point in~$\mypath$. The via-path oracle we
presented above is not suitable for that: it only records the first point on~$\mypath$
that can be reached from~$s$ and the last point from which we reach reach~$t$,
and going via these points may take many more hops than a shortest 
path via~$\mypath$ would. To overcome this problem, we will store
the first point on~$\mypath$ that can be reached from~$s$
(and the last point that can reach $t$) \emph{with a path of a given length}. 
To avoid increasing the storage too much, we will not do this for all 
possible lengths, but only for logarithmically many lengths. More precisely,
our via-path oracle stores, for all~$p\in P$ and $j\in\{-1,\ldots,\ceil{\log_{1+\eps} n}\}$, the values
\[
\MinIn_{\mypath}[p,j] := \min \{ i : \dhop(p,q_i) \leq (1+\eps)^j \mbox{ and } q_i \in \mypath \}
\]
and
\[ 
\MaxOut_{\mypath}[p,j] := \max \{ i : \dhop(q_i,j) \leq (1+\eps)^j \mbox{ and } q_i \in \mypath \}.
\]
Here we include $j=-1$, since for $p\in\mypath$ there is a path of zero hops to (resp.~from)~$\mypath$,
namely to (resp.~from)~$p$ itself.
The next lemma shows that we can get a $(1+\eps)$-approximation of $d_{\mypath}(s,t)$
using the arrays $\MinIn_{\mypath}$ and $\MaxOut_{\mypath}$, up to an additional additive error of a single hop.
%-------------------------------------------------------------------------------------- 
\begin{lemma}\label{lem:distance-vp-oracle}
Let $d^*_{\mypath}(s,t) := \min \{ (1+\eps)^j + (1+\eps)^k : \MinIn_{\mypath}[s,j] \leq \MaxOut_{\mypath}[t,k]  \}$.
If $d_{\mypath}(s,t)=\infty$ then $d^*_{\mypath}(s,t)=\infty$, and otherwise
$d_{\mypath}(s,t) \leq d^*_{\mypath}(s,t) < (1+\eps) \cdot d_{\mypath}(s,t)+1$.
\end{lemma}
%-------------------------------------------------------------------------------------- 
\begin{proof}
First note that for any pair $j,k$ with $\MinIn_{\mypath}[s,j] \leq \MaxOut_{\mypath}[t,k]$,
there is an arc from $q_{a}$ to $q_b$ for $a:=\MinIn_{\mypath}[s,j]$ and $b:=\MaxOut_{\mypath}[t,k]$.
Hence,
\begin{equation} \label{eq:jk}
\mbox{if} \hspace*{5mm} \MinIn_{\mypath}[s,j] \leq \MaxOut_{\mypath}[t,k] \hspace*{5mm} \mbox{then} \hspace*{5mm} d_{\mypath}(s,t) \leq (1+\eps)^j + (1+\eps)^k +1.
\end{equation}

To prove the first part of the lemma, suppose $d_{\mypath}(s,t)=\infty$.
Then we must also have $d^*_{\mypath}(s,t) = \infty$, as claimed.
Indeed, if $d^*_{\mypath}(s,t) < \infty$ there are $j,k$ 
with $\MinIn_{\mypath}[s,j] \leq \MaxOut_{\mypath}[t,k]$ which,
together with~(\ref{eq:jk}), would contradict~$d_{\mypath}(s,t)=\infty$. 

To prove the second part of the lemma, suppose $d_{\mypath}(s,t)\neq\infty$. Then~(\ref{eq:jk}) and the definition
of $d^*_{\mypath}(s,t)$ immediately imply that $d_{\mypath}(s,t) \leq d^*_{\mypath}(s,t)$.
To prove $d^*_{\mypath}(s,t) \leq (1+\eps) \cdot d_{\mypath}(s,t)+1$, let
$q_{i^*}\in\mypath$ be a point such that $d_{\mypath}(s,t) = \dhop(s,q_{i^*})+\dhop(q_{i^*},t)$.
Define
\[
j^* := \min \big\{ j : -1\leq j\leq \ceil{\log_{1+\eps} n} \mbox{ and } \dhop(s,q_{i^*}) \leq (1+\eps)^{j} \big\}
\]
and
\[
k^* := \min \big\{ k : -1\leq k\leq \ceil{\log_{1+\eps} n} \mbox{ and } \dhop(s,q_{i^*}) \leq (1+\eps)^{k} \big\}.
\]
Then 
\[
\MinIn_{\mypath}[s,j^*] \leq i^* \leq \MaxOut_{\mypath}[t,k^*]
\]
and so $d^*_{\mypath}(s,t) \leq  (1+\eps)^{j^*} + (1+\eps)^{k^*} + 1$ by~(\ref{eq:jk}). Moreover,
\[
d_{\mypath}(s,t) = \dhop(s,q_{i^*})+\dhop(q_{i^*},t) > (1+\eps)^{j^*-1} + (1+\eps)^{k^*-1}.
\]
Hence,
\[
d^*_{\mypath}(s,t) \leq  (1+\eps)^{j^*} + (1+\eps)^{k^*} + 1
                   < (1+\eps) \cdot d_{\mypath}(s,t)+1,
\]
thus finishing the proof. 
\end{proof}
%-------------------------------------------------------------------------------------
The following lemma summarizes the performance of our via-path approximate distance
oracle.
%-------------------------------------------------------------------------------------- 
\begin{lemma}\label{lem:distance-vp-oracle-performance}
Let $\tg(P)$ be a transmission graph on $n$ points, and let $\mypath$ be  transitive path in $\tg(P)$.
Then there is an oracle that uses $O((1/\eps)\cdot n\log n)$ storage such that, for two query points
$s,t\in P$ we can compute the value $d^*_{\mypath}(s,t)$ from Lemma~\ref{lem:distance-vp-oracle}
in $O((1/\eps)\log n)$ time.
\end{lemma}
%-------------------------------------------------------------------------------------- 
\begin{proof}
The storage needed for the arrays $\MinIn$ and $\MaxOut$ is $O(n \log_{1+\eps} n)$.
Since $\log_{1+\eps}n = (\log n)/\log (1+\eps)= O((1/\eps)\log n)$ this proves the
bound on the amount of storage of our oracles. 

To compute $d^*_{\mypath}(s,t)$ when answering a query,
we note that the values $\MinIn[s,j]$ are non-increasing as~$j$ increases, and the values
$\MaxOut[t,k]$ are non-decreasing as~$k$ increases. Hence, we can find 
$d^*_{\mypath}(s,t) = \min \{ (1+\eps)^j + (1+\eps)^k : \MinIn_{\mypath}[s,j] \leq \MaxOut_{\mypath}[t,k]  \}$
by scanning the rows $\MinIn[s,\cdot]$ and $\MaxOut[t,\cdot]$, as follows.
First, for $j=\ceil{\log_{1+\eps} n}$ we find the smallest $k$ such that
$\MinIn_{\mypath}[s,j] \leq \MaxOut_{\mypath}[t,k]$. If no such $k$ exists  
then $d^*_{\mypath}(s,t)=\infty$ and we are done. Otherwise we decrease $j$ by~1
and increase~$k$ until we again have $\MinIn_{\mypath}[s,j] \leq \MaxOut_{\mypath}[t,k]$.
We continue scanning $\MinIn[s,\cdot]$ and $\MaxOut[t,\cdot]$ in opposite directions,
until we have found for each $j$ the smallest $k$ such that $\MinIn_{\mypath}[s,j] \leq \MaxOut_{\mypath}[t,k]$.
One of the pairs $j,k$ thus found must be the pair defining~$d^*_{\mypath}(s,t)$.
Thus the query time is linear in the length of the rows $\MinIn[s,\cdot]$ and $\MaxOut[t,\cdot]$,
which is $O(\log_{1+\eps} n)= O((1/\eps)\log n)$.
\end{proof}
%-------------------------------------------------------------------------------------- 
The via-path oracle for approximate distance queries can be extended to a via-clique
oracle, in exactly the same way as was done for reachability queries: cover each clique
by $O(1)$ stabbed cliques, cover each stabbed clique by at most six (transitive) paths,
and construct a via-path oracle for each of them.
These via-clique oracles can then be plugged into the separator approach, again
exactly as before. 

To answer a query with query points $s,t$ on our complete oracle, we first 
compute an approximate hop-distance from $s$ to $t$ via one of the cliques in the clique-based
separator. If $s$ and $t$ are in the same part of the partition, we then also
recursively approximate the hop-distance within that part. Finally, we return
the minimum of the distances found.
\medskip

The analysis of the amount of storage and query time is as before; the only difference is 
that the via-clique oracle has an extra factor $O((1/\eps)\log n)$ in the storage
and query time, which carries over to the final bounds. 
Unfortunately, the preprocessing time does not carry over. The reason
is that we cannot afford to work with a spanner. It is possible to do 
a BFS on $\tg(P)$ efficiently, as demonstrated by Kaplan~\etal~\cite{DBLP:journals/siamcomp/KaplanMRS18}
and by An and Oh~\cite{An-Oh-Algorithmica}, but it is not quite clear how to
integrate this into the preprocessing algorithm. Moreover, 
we also need to do a BFS on the ``reversed'' transmission graph, which makes
things even more difficult.
We therefore revert to simply solving the All-Pairs Shortest Path problem (APSP)
in a preprocessing step, thus giving all pairwise distances. Once these distances
are available, the construction of the oracle can be done within
the same time bounds as before.
Solving APSP in directed unweighted graphs can be done in $O(n^{2.5302})$ time~\cite{DBLP:conf/focs/Gall12}.
We can improve the preprocessing time by observing that, since we are
approximating the hop-distance anyway, we may as well solve APSP approximately.
Indeed, if we work with $(1+\eps)$-approximations of the distances,
then the total (multiplicative) approximation factor of our oracle is bounded by $(1+\eps)^2<1+3\eps$.
Hence, by working with parameter~$\eps':=\eps/3$, we can get a $(1+\eps')$-approximation, for any $\eps'>0$.
Galil and Margalit~\cite{DBLP:journals/iandc/GalilM97} (see also the paper by Zwick~\cite{DBLP:journals/jacm/Zwick02}) have shown that
one can compute a $(1+\eps)$-approximation of all pairwise distances 
in $\tilde{O}(n^{\omega}/\eps)$ time,\footnote{The notation $\tilde{O}$ hides $O(\polylog n)$ factors.}
where $\omega < 2.373$ is the matrix-multiplication exponent~\cite{DBLP:conf/soda/AlmanW21}.
\medskip

Putting everything together, we obtain the following theorem.
%-------------------------------------------------------------------------------------- 
\begin{theorem}\label{thm:distance-oracle}
Let $P$ be a set of $n$ points in~$\Reals^2$. Then there is an oracle for approximate
distance queries that uses $O((n/\eps)\sqrt{n}\log n)$ storage and that, given two query 
points~$s,t\in P$ can determine if $s$ can reach $t$. If so, it can report a value $d^*_{\mypath}(s,t)$ with
$d_{\mypath}(s,t) \leq d^*_{\mypath}(s,t) \leq (1+\eps) \cdot d_{\mypath}(s,t)+1$
in $O((\sqrt{n}/\eps)\log n)$ time. The oracle can be constructed in $\tilde{O}(n^{\omega}/\eps)$ time,
where $\omega < 2.373$ is the matrix-multiplication exponent.
\end{theorem}
%-------------------------------------------------------------------------------------- 

%-------------------------------------------------------------------------------------- 
\subparagraph{Extension to continuous reachability queries.}
%-------------------------------------------------------------------------------------- 
In a continuous reachability query, we are given a query pair $s,t\in P\times \Reals^2$
and we wish to decide if $s$ can reach~$t$, that is, we wish to decide if there is point $q\in P$
such that $s\reaches q$ and $|qt|\leq r(q)$. We define the hop-distance from $s$ to $t$,
denoted by $\dhop(s,t)$, as
\[
\dhop(s,t) = \min \{ \dhop(s,q) + 1 : q\in P \mbox{ and } |qt| \leq r(q) \}.
\]
Next we present a data structure that, given a query target point~$t\in\Reals^2$,
determines a set $Q(t)$ of at most six points that can serve as the last point
on a path from any source point $s$ to $t$. To this end,
let $\Gamma_1(t),\ldots,\Gamma_6(t)$ be the six canonical cones of~$t$.
(See Fig.~\ref{fig:cones}(i) for an illustration of the concept of canonical cones.)
For each cone~$\Gamma_i(t)$, let $Q_i(t)\subseteq P$ be the set of points in $\Gamma_i(t)$ whose transmission
disk contains~$t$. (If a point $p$ with $t\in D_p$ lies on the boundary
between two cones, we assign it to one of them arbitrarily.)
For $Q_i(t)\neq\emptyset$, 
let $q_i(t)$ be a point in $Q_i(t)$ of minimum radius. Define $Q(t)$ to be set
of at most six points $q_i(t)$ selected in this manner. 
%-------------------------------------------------------------------------------------- 
\begin{observation}\label{obs:Qt}
Let~$s\in P$ be a point that can reach $t$ and let ${d^*(s,t) := \min_{q\in Q(t)}\dhop(s,q)+1}$. Then
$\dhop(s,t) \leq d^*(s,t) \leq \dhop(s,t)+1$.
\end{observation}
%-------------------------------------------------------------------------------------- 
\begin{proof}
Suppose $s$ can reach $t$, and let $q^*\in P$ be the point immediately preceding~$t$
on a shortest path from $s$ to~$t$. Let $i$ be such that $q^*\in Q_i(t)$, and
let $q_i(t) \in Q_i(t)$ be the minimum-radius point that was put into~$Q(t)$. 
Then $q^*$ and $q_i(t)$ lie in the same 60-degree cone~$\Gamma_i(t)$
and $r(q_i(t)) \leq r(q^*)$. Moreover, $|q^*t|\leq r(q^*)$ and  $|q_i(t)t|\leq r(q_i(t))$.
Hence, $(q^*,q_i(t))$ is an arc in~$\tg(P)$---this was also the key property
underlying the via-path oracle---and so $\dhop(s,q_i(t)) \leq  \dhop(s,q^*) + 1$. This implies 
\[
d^*(s,t) \leq \dhop(s,q_i(t)) + 1  \leq  \dhop(s,q^*) + 2  = \dhop(s,t) +1.
\]
On the other hand, any point $q\in Q(t)$ can reach $t$ with one hop, so
\[
\dhop(s,t) \leq  \min_{q\in Q(t)}\dhop(s,q) +1 =d^*(s,t).
\]
\end{proof}
%-------------------------------------------------------------------------------------- 
We now describe a data structure that, given a query point~$t$, 
computes the point~$q_i(t)$ mentioned in Observation~\ref{obs:Qt} (if it exists),
for a fixed index~$i$.
We will construct such a data structure for each $1\leq i \leq 6$, and by querying each
of these six structures we can compute the set $Q(t)$.

Consider a fixed index~$i$, that is, consider the canonical cone of a fixed orientation.
The structure for this fixed index~$i$ is defined as follows. 
For a point $p\in P$, let $\overline{\Gamma}_i(p)$ be the canonical cone opposite $\Gamma_i(p)$---for instance,
$\overline{\Gamma}_1(p) = \Gamma_4(p)$---and define $D^{(i)}_p := D_p \cap \overline{\Gamma}_i(p)$.
Then $p\in Q_i(t) \mbox{ if and only if } t\in  D^{(i)}_p$; see Fig.~\ref{fig:continuous-query}(i).
%--------------------------------------------------------------------------------------
\begin{figure}
\begin{center}
\includegraphics{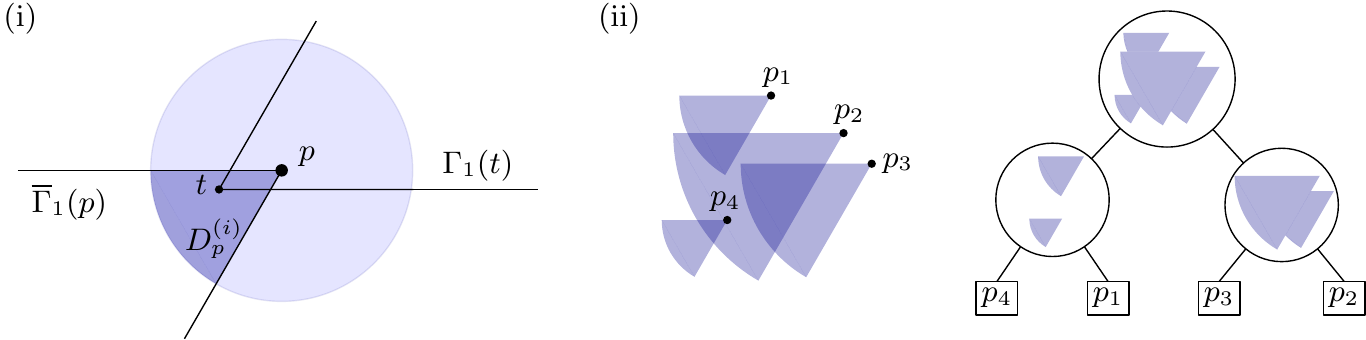}
\end{center}
\caption{(i) Illustration of the fact that $p\in Q_i(t)$ if and only if $t\in D_p^{(i)}$.
         (ii) The set $\D^*$ and the corresponding search structure.}
\label{fig:continuous-query}
\end{figure}
%--------------------------------------------------------------------------------------
Thus we need a data structure on the set $\D^* := \{ D^{(i)}_p : p\in P\}$
that, given a query point~$t\in\Reals^2$, can quickly find a point
$p\in P$ with minimum transmission radius such that $t\in  D^{(i)}_p$.
To this end we construct a balanced binary tree~$\T_i$, as follows;
see Fig.~\ref{fig:continuous-query}(ii).
\begin{itemize}
\item The leaves of $\T_i$ store the points $p\in P$ in left-to-right order according 
      to their transmission radii. For example, the leftmost leaf stores a point $p\in P$ of 
      minimum transmission radius and the rightmost leaf stores a point $p\in P$ of 
      maximum transmission radius.
\item Let $P(v)$ denote the set of points stored in the subtree rooted at~$v$.
     Then $v$ stores the union $\U(v) := \bigcup_{p\in P(v)} D^{(i)}_p$, preprocessed such
     that for a query point $t\in \Reals^2$ we can decide if $t\in \U(v)$ in $O(\log |\U(v)|)$ time. 
     Using a standard point-location data structure~\cite{DBLP:reference/cg/Snoeyink04}
     this associated structure needs $(|\U(v)|)$ storage.     
\end{itemize}
A query with point~$t\in\Reals^2$ to find $q_i(t)$ is answered as follows. 
If $t\not\in \U(\myroot(\T_i))$ then $Q_i(t)=\emptyset$ and we are done. Otherwise
we descend in~$\T_i$, proceeding to the left child $v_{\mathrm{left}}$
of the current node~$v$ if $t\in \U(v_{\mathrm{left}})$, and proceeding to the right child of~$v$
otherwise. We then report the point stored in the leaf where the search ends.
It is easy to see that this correctly answers the query. 
This leads to the following theorem.
%-------------------------------------------------------------------------------------- 
\begin{theorem}
Let $P$ be a set of $n$ points in $\Reals^2$. There is a data structure of size $O(n\log n)$
that, for any query point $t\in\Reals^2$, can find in $O(\log^2 n)$ time a set $Q(t)\subset P$ 
of at most six points with the following property:  for any point~$s\in P$ 
 we have $\dhop(s,t) \leq d^*(s,t) \leq \dhop(s,t)+1$,
where ${d^*(s,t) := \min_{q\in Q(t)}\dhop(s,q)+1}$. 
\end{theorem}
%-------------------------------------------------------------------------------------- 
\begin{proof} 
Observation~\ref{obs:Qt} states that the set $Q(t)$ defined earlier has the desired properties.
It follows from the discussion above that we can find $Q(t)$ by querying 
the six structures $\T_i$ as described above. It remains to argue that each $\T_i$
uses $O(n\log n)$ storage and has $O(\log^2 n)$ query time.

The total amount of storage is $\sum_{v\in\T_i} O(u(v))$, where $u(v)$ is the complexity of the union~$\U(v)$. 
Recall that each $D^{(i)}_p$ is a 60-degree sector of a disk, and that all these sectors
have exactly the same orientation. Hence, $\U(v)$ is the union of a set of 
\emph{homothets}---they are scaled and translated copies of each other---and 
so their union complexity is linear~\cite{state-of-union}.
Thus $u(v) =O(|P(v)|)$, from which it follows by standard arguments 
that the total amount of storage is $O(n\log n)$.
The query time is $O(\log^2 n)$, since we spend $O(\log n)$ per node as we descend~$\T_i$
and the depth of $\T_i$ is $O(\log n)$. 
\end{proof}
%-------------------------------------------------------------------------------------- 
\begin{corollary}
Let $P$ be a set of $n$ points in~$\Reals^2$, each with an associated transmission radius,
and let $\tg(P)$ be the corresponding the transmission graph.
\begin{enumerate}[(i)]
\item There is an oracle for continuous reachability queries in~$\tg(P)$
      that has $O(\sqrt{n})$ query time and uses $O(n\sqrt{n})$ storage. 
\item There is an oracle for continuous approximate-distance queries in~$\tg(P)$
       that uses $O((n/\eps)\sqrt{n}\log n)$ storage and that, given query 
       points~$s\in P$ and $t\in\Reals^2$ can report a value $d^*_{\mypath}(s,t)$ with
       $d_{\mypath}(s,t) \leq d^*_{\mypath}(s,t) \leq (1+\eps) \cdot d_{\mypath}(s,t)+2$ 
       in $O((\sqrt{n}/\eps)\log n)$ time.
\end{enumerate}
\end{corollary}
%-------------------------------------------------------------------------------------- 

%% file: conclusion.tex
%--------------------------------------------------------------------------------------
\section{Concluding remarks}
\label{sec:conclusion}
%--------------------------------------------------------------------------------------
We presented oracles for reachability and approximate distance queries in transmission
graphs, whose performance does not depend on~$\Psi$, the ratio between the largest
and smallest transmission radius. The bounds we obtain are a significantly improvement
over the previously best known bounds by An and Oh~\cite{An-Oh-Algorithmica}. 
However, our bounds for reachability queries  are still quite far from what can be achieved when~$\Psi<\sqrt{3}$: 
we obtain $O(\sqrt{n})$ query time with $O(n\sqrt{n})$ storage, while for $\Psi<\sqrt{3}$ 
Kaplan~\etal~\cite{DBLP:journals/siamcomp/KaplanMRS18} obtain $O(1)$ query time using $O(n)$ storage.
They do this by reducing the problem to one on planar graphs. 
For $\Psi>\sqrt{3}$ such a reduction seems quite hard, if not impossible, and to make progress
a much deeper understanding of the structure of transmission graphs may be needed.
A first goal could be to improve on our bounds for small values of $\Psi$, say $\Psi=2$.
Lower bounds for unbounded~$\Psi$ would also be quite interesting.

%% file: main.bbl
\begin{thebibliography}{10}

\bibitem{state-of-union}
Pankaj~K. Agarwal, Janos Pach, and Micha Sharir.
\newblock State of the union (of geometric objects): A review.
\newblock In {\em Computational Geometry: Twenty Years Later}, pages 9--48.
  American Mathematical Society, 2008.

\bibitem{DBLP:conf/soda/AlmanW21}
Josh Alman and Virginia~Vassilevska Williams.
\newblock A refined laser method and faster matrix multiplication.
\newblock In {\em Proc.~2021 {ACM-SIAM} Symposium on Discrete Algorithms
  ({SODA})}, pages 522--539, 2021.

\bibitem{An-Oh-Algorithmica}
Shinwoo An and Eunjin Oh.
\newblock Reachability problems for transmission graphs.
\newblock {\em Algorithmica}, 84:2820--2841, 2022.

\bibitem{DBLP:conf/esa/ArikatiCCDSZ96}
Srinivasa~Rao Arikati, Danny~Z. Chen, L.~Paul Chew, Gautam Das, Michiel H.~M.
  Smid, and Christos~D. Zaroliagis.
\newblock Planar spanners and approximate shortest path queries among obstacles
  in the plane.
\newblock In {\em Proc.~4th Annual European Symposium on Algorithms ({ESA})},
  volume 1136 of {\em Lecture Notes in Computer Science}, pages 514--528, 1996.

\bibitem{DBLP:conf/stoc/ChenX00}
Danny~Z. Chen and Jinhui Xu.
\newblock Shortest path queries in planar graphs.
\newblock In {\em Proc.~32nd Annual {ACM} Symposium on Theory of Computing
  ({STOC})}, pages 469--478, 2000.

\bibitem{bbkmz-ethf-20}
Mark de~Berg, Hans~L. Bodlaender, S{\'{a}}ndor Kisfaludi{-}Bak, D{\'{a}}niel
  Marx, and Tom~C. van~der Zanden.
\newblock A framework for {Exponential-Time-Hypothesis}-tight algorithms and
  lower bounds in geometric intersection graphs.
\newblock {\em SIAM J. Comput.}, 49:1291--1331, 2020.

\bibitem{Djidjev1996}
Hristo~N. Djidjev.
\newblock Efficient algorithms for shortest path queries in planar digraphs.
\newblock In {\em Proc.~22nd International Workshop on Graph-Theoretic Concepts
  in Computer Science}, volume 1197 of {\em Lecture Notes in Computer Science},
  1997.

\bibitem{DBLP:journals/siamcomp/Frederickson87}
Greg~N. Frederickson.
\newblock Fast algorithms for shortest paths in planar graphs, with
  applications.
\newblock {\em {SIAM} J. Comput.}, 16(6):1004--1022, 1987.

\bibitem{DBLP:journals/iandc/GalilM97}
Zvi Galil and Oded Margalit.
\newblock All pairs shortest distances for graphs with small integer length
  edges.
\newblock {\em Inf. Comput.}, 134(2):103--139, 1997.

\bibitem{DBLP:conf/focs/Gall12}
Fran{\c{c}}ois~Le Gall.
\newblock Faster algorithms for rectangular matrix multiplication.
\newblock In {\em Proc.~53rd Annual {IEEE} Symposium on Foundations of Computer
  Science ({FOCS})}, pages 514--523, 2012.

\bibitem{DBLP:conf/soda/GawrychowskiMWW18}
Pawel Gawrychowski, Shay Mozes, Oren Weimann, and Christian Wulff{-}Nilsen.
\newblock Better tradeoffs for exact distance oracles in planar graphs.
\newblock In {\em Proc.~29th Annual {ACM-SIAM} Symposium on Discrete Algorithms
  ({SODA})}, pages 515--529, 2018.

\bibitem{DBLP:journals/jacm/Har-PeledR15}
Sariel Har{-}Peled and Benjamin Raichel.
\newblock Net and prune: {A} linear time algorithm for euclidean distance
  problems.
\newblock {\em J. {ACM}}, 62(6):44:1--44:35, 2015.

\bibitem{DBLP:conf/focs/HolmRT15}
Jacob Holm, Eva Rotenberg, and Mikkel Thorup.
\newblock Planar reachability in linear space and constant time.
\newblock In {\em Proc.~56th {IEEE} Annual Symposium on Foundations of Computer
  Science ({FOCS})}, pages 370--389, 2015.

\bibitem{DBLP:journals/siamcomp/KaplanMRS18}
Haim Kaplan, Wolfgang Mulzer, Liam Roditty, and Paul Seiferth.
\newblock Spanners for directed transmission graphs.
\newblock {\em {SIAM} J. Comput.}, 47(4):1585--1609, 2018.

\bibitem{DBLP:journals/algorithmica/KaplanMRS20}
Haim Kaplan, Wolfgang Mulzer, Liam Roditty, and Paul Seiferth.
\newblock Reachability oracles for directed transmission graphs.
\newblock {\em Algorithmica}, 82(5):1259--1276, 2020.

\bibitem{DBLP:conf/focs/LeW21}
Hung Le and Christian Wulff{-}Nilsen.
\newblock Optimal approximate distance oracle for planar graphs.
\newblock In {\em Proc.~62nd {IEEE} Annual Symposium on Foundations of Computer
  Science ({FOCS})}, pages 363--374, 2021.

\bibitem{DBLP:journals/dcg/MatousekW04}
Jir{\'{\i}} Matousek and Uli Wagner.
\newblock New constructions of weak epsilon-nets.
\newblock {\em Discret. Comput. Geom.}, 32(2):195--206, 2004.

\bibitem{MTTV-sep-sphere-packing}
Gary~L. Miller, Shang{-}Hua Teng, William~P. Thurston, and Stephen~A. Vavasis.
\newblock Separators for sphere-packings and nearest neighbor graphs.
\newblock {\em J. {ACM}}, 44(1):1--29, 1997.

\bibitem{SW-geom-sep}
Warren~D. Smith and Nicholas~C. Wormald.
\newblock Geometric separator theorems {\&} applications.
\newblock In {\em Proc.~39th {IEEE} Annual Symposium on Foundations of Computer
  Science ({FOCS})}, pages 232--243, 1998.

\bibitem{DBLP:reference/cg/Snoeyink04}
Jack Snoeyink.
\newblock Point location.
\newblock In Jacob~E. Goodman and Joseph O'Rourke, editors, {\em Handbook of
  Discrete and Computational Geometry, Second Edition}, pages 767--785. Chapman
  and Hall/CRC, 2004.

\bibitem{DBLP:conf/soda/Wulff-Nilsen16}
Christian Wulff{-}Nilsen.
\newblock Approximate distance oracles for planar graphs with improved query
  time-space tradeoff.
\newblock In {\em Proc.~27th Annual {ACM-SIAM} Symposium on Discrete Algorithms
  ({SODA})}, pages 351--362, 2016.

\bibitem{DBLP:journals/jacm/Zwick02}
Uri Zwick.
\newblock All pairs shortest paths using bridging sets and rectangular matrix
  multiplication.
\newblock {\em J. {ACM}}, 49(3):289--317, 2002.

\end{thebibliography}
